\documentclass{article}
\usepackage{latexsym}
\newtheorem{theorem}{Theorem}

\newcommand{\blackslug}{\mbox{\hskip 1pt \vrule width 4pt height 8pt 
depth 1.5pt \hskip 1pt}}
\newcommand{\qed}{\quad\blackslug\lower 8.5pt\null\par\noindent}
\newenvironment{proof}{\par\noindent{\bf Proof:}}{\qed \par}

\title{Representing choice functions by a total hyper-order}
\author{Daniel Lehmann\\
The Rachel and Selim Benin School \\of Computer Science and Engineering, \\Hebrew University, 
\\Jerusalem 91904, Israel
\\lehmann@cs.huji.ac.il
}
\date{July 2021} 
\begin{document}
\maketitle

\begin{abstract}
Choice functions over a set $X$ that satisfy the Outcast, a.k.a. Aizerman, property are exactly 
those that attach to any set its maximal subset relative to some total order of ${2}^{X}$.
\end{abstract}

\section{Definitions}
A base set $X$, finite or infinite, is assumed.
In this paper we shall call a total order $\leq$ on ${2}^{X}$, also called a hyper-order, 
{\em a well-order} if there is no infinite {\em ascending} chain. 

Any total well-order $\leq$ defines a function 
\mbox{$f_{\leq} : {2}^{X} \longrightarrow {2}^{X}$}, by;
\[
f_{\leq} ( A ) = \max \{ B \subseteq X \mid B \subseteq A \}
\]
for any \mbox{$A \subseteq X$}.

A function \mbox{$f : {2}^{X} \longrightarrow {2}^{X}$} is said to be a {\em choice function}
iff \mbox{$f ( A ) \subseteq A$} for any \mbox{$A \subseteq X$}.
A choice function $f$ is said to satisfy {\em Outcast}, a.k.a. Aizerman~\cite{AizerMalish:81} 
iff, for any \mbox{$A , B \subseteq X$}
\[
f (A ) \subseteq B \subseteq A {\rm \ implies \ } f ( A ) = f ( B ).
\]

\section{Representation result}
\begin{theorem} \label{the:rep}
A function \mbox{$f : {2}^{X} \longrightarrow {2}^{X}$} is a choice function that satisfies
Outcast iff it is the function $f_{\leq}$ defined by some total well-order on ${2}^{X}$.
\end{theorem}
\begin{proof}
For the {\em if} direction, note that, by definition, \mbox{$f_{\leq} ( A) ) \subseteq A$}
and that \mbox{$B \subseteq A$} implies \mbox{$f_{\leq} ( B ) \leq f_{\leq} ( A ) $}, but
\mbox{$ f_{ \leq } ( A ) \subseteq B$} implies 
\mbox{$ f_{\leq} ( A )  \leq f_{\leq} ( B)$}.

For the {\em only if} direction, assume $f$ is a choice function that satisfies Outcast.
We say that \mbox{$A \subseteq X$} is a {\em fixpoint} if \mbox{$f ( A ) = A$}.
By Outcast, for any $A$, $f ( A ) $ is a fixpoint.
To any fixpoint $A$, we associate its {\em domain} 
\mbox{$\bar{A} = \{ B \subseteq X \mid f ( B ) = A \}$}. 
Note that if \mbox{$C \in \bar{A} \cap \bar{B}$}, then \mbox{$f ( C ) = A$}, 
\mbox{$f ( C ) = B$} and \mbox{$A = B$}.
We see that $A$ is the only fixpoint in $\bar{A}$.
For any domain $D$ let $D_{h}$ be the only fixpoint in $D$.
We can therefore order any domain $D$ by a well-order $\leq_{D}$ such that 
\mbox{$D_{h} = \max ( D ) $}.
Let $\cal F$ be the set all fixpoints of $f$.
The set $\cal F$ is equipped by the following partial order:
\mbox{$A \leq B$} iff \mbox{$A \subseteq B$}.
We can extend it to a total well-order on $\cal F$, denoted also $\leq$.
The total well-orders $\leq$ on $\cal F$ and $\leq_{D}$ on the different domains
to a total well-order on ${2}^{X}$:
\begin{itemize}
\item if \mbox{$D_{ f ( B ) } < D_{ f ( A ) }$} we put \mbox{$B \leq A$},
\item if \mbox{$D_{ f ( A ) } < D_{ f ( B ) }$} we put \mbox{$A \leq B$},
\item if \mbox{$\overline { f ( A ) } = \overline { f ( B ) } $} 
\begin{itemize}
\item if \mbox{$A \leq_{\overline{ f ( A ) }} B$}, we put \mbox{$A \leq B$}, and
\item if \mbox{$B \leq_{\overline{ f ( A ) }} A$}, we put \mbox{$B \leq A$}.
\end{itemize}
\end{itemize}
The resulting relation $\leq$ is a total well-order.
Note that, for every \mbox{$A \subseteq X$}, \mbox{$A \leq f ( A )$} since
\mbox{$f ( f ( A ) ) = f ( A )$}.
Note also that for any \mbox{$ B \subseteq A \subseteq X$} 
one has \mbox{$ B \leq f ( A ) $}:
indeed 
\begin{itemize}
\item if \mbox{$ f ( B ) \leq f ( A )$}, \mbox{$ B \leq f ( B ) \leq f ( A )$},
\item if \mbox{$ f ( B ) \not \leq f ( A ) $}, then, by the definition of $\leq$ on $\cal F$,  
\mbox{$ B \not \subseteq A$}, a contradiction.
\end{itemize}
 We conclude that \mbox{$f ( A ) = f_{\leq} ( A ) $}.
\end{proof}

\bibliographystyle{plain}

\end{document}